\documentclass[12pt,onecolumn]{IEEEtran}
\ifCLASSINFOpdf
\else
\fi

\usepackage{amsmath,amsthm}
\theoremstyle{plain}
\usepackage{color}

\newtheorem{thm}{\textbf{Theorem}}[]

\usepackage{xpatch}
\makeatletter
\xpatchcmd{\@thm}{\thm@headpunct{.}}{\thm@headpunct{}}{}{}
\makeatother
\makeatletter

\renewenvironment{proof}[1][\proofname]{\par
\pushQED{\qed}%
\normalfont 
\trivlist
\item\relax
{\itshape
#1\@addpunct{:}}\hspace\labelsep\ignorespaces
}{%
\popQED\endtrivlist\@endpefalse
}
\AtBeginEnvironment{thm}{}
	\AtEndEnvironment{thm}{}

\renewenvironment{proof}[1][\proofname]{\par
\pushQED{\qed}%
\normalfont 
\trivlist
\item\relax
{\itshape
#1\@addpunct{:}}\hspace\labelsep\ignorespaces
}{%
\popQED\endtrivlist\@endpefalse
}
\AtBeginEnvironment{corr}{}
	\AtEndEnvironment{corr}{}

\usepackage{amsmath,amssymb,latexsym,float,epsfig,epstopdf,subfigure,graphicx,array,algorithm,cite,booktabs}
 \usepackage{multirow}

\usepackage{amsmath,cuted,amssymb,latexsym,physics,float,epsfig,subfigure,epstopdf,graphicx,array,algorithm,algorithmic,makecell,url,setspace}
\usepackage{epsf}
\usepackage{float}
\usepackage{stfloats}
\usepackage{tikz}
\usetikzlibrary{positioning,shapes,arrows}
\usetikzlibrary{matrix}
\usetikzlibrary{decorations.pathreplacing}

\tikzstyle{overbrace text style}=[above, pos=.5, yshift=3mm]
\tikzstyle{overbrace style}=[decorate,decoration={brace,raise=2mm,amplitude=3pt}]
\tikzstyle{underbrace style}=[decorate,decoration={brace,raise=2mm,amplitude=3pt,mirror},color=gray]
\tikzstyle{underbrace text style}=[font=\tiny, below, pos=.5, yshift=-3mm]

\usepackage[bottom]{footmisc}
\let\svfootnoterule\footnoterule
\renewcommand\footnoterule{\vfill\svfootnoterule}
\usepackage{lipsum}

\newtheorem{lemma}{Lemma}

\newtheorem{remark}{Remark}
\newtheorem{corr}{Corollary}
\newtheorem*{prop}{Proposition}

\newmuskip\pFqmuskip

\newcommand*\pFq[6][8]{%
  \begingroup 
  \pFqmuskip=#1mu\relax
  \mathcode`\,=\string"8000
  \begingroup\lccode`\~=`\,
  \lowercase{\endgroup\let~}\pFqcomma
  {}_{#2}F_{#3}{\left[\genfrac..{0pt}{}{#4}{#5};#6\right]}%
  \endgroup
}
\newcommand{\pFqcomma}{\mskip\pFqmuskip}

\begin{document}

\tikzstyle{block} = [draw, fill=blue!20, rectangle,
    minimum height=2.5em, minimum width=2.5em]
\tikzstyle{sum} = [draw, fill=blue!20, circle, node distance=1cm]
\tikzstyle{input} = [coordinate]
\tikzstyle{output} = [coordinate]
\tikzstyle{pinstyle} = [pin edge={to-,thin,black}]

\newenvironment{nscenter}
 {\parskip=0pt\par\nopagebreak\centering}
 {\par\noindent\ignorespacesafterend}

\doublespacing

%
\title{Non-Orthogonal Multiple Access for Visible Light Communications with Ambient Light and User Mobility}
%
%
%

\author{Rangeet~Mitra, Member, IEEE, Paschalis C. ~Sofotasios, Senior Member, IEEE, Vimal~Bhatia,
     Senior Member, IEEE,
 and Sami~Muhaidat, Senior Member, IEEE       

\thanks{R. Mitra is with the Indian Institute of Information Technology SriCity-517646, India, (email: rangeet.mitra@iiits.in).}

\thanks{P. C. Sofotasios is with the Center for Cyber-Physical Systems, Department of Electrical Engineering and Computer Science, Khalifa University of Science and Technology, PO Box 127788, Abu Dhabi, UAE, and also with the Department of Electrical Engineering, Tampere University, 33101 Tampere, Finland (email: p.sofotasios@ieee.org).}

\thanks{V. Bhatia is with the Department of Electrical Engineering, Indian Institute of Technology Indore, Indore-453441, India, (email: vbhatia@iiti.ac.in).}

\thanks{S. Muhaidat is with the Center for Cyber-Physical Systems, Department of Electrical Engineering and Computer Science, Khalifa University of Science and
Technology, PO Box 127788, Abu Dhabi, UAE (email: muhaidat@ieee.org).}
}
\maketitle

\begin{abstract}
The ever-increasing demand for high data-rate applications and the proliferation of connected devices pose several theoretical and technological challenges  for the  fifth generation (5G) networks and beyond. Among others, this includes the  spectrum scarcity and massive connectivity of devices, particularly in the context of the Internet of Things (IoT) ecosystem. In this respect, visible light communication (VLC) has recently emerged as a potential solution  for these challenges, particularly in scenarios relating to indoor communications. Additionally, non-orthogonal multiple access (NOMA) for VLC has been envisioned to address some of the key challenges in the next generation wireless networks. However, in realistic environments, it has been shown that  VLC systems suffer from  additive optical interference due to ambient light, and user-mobility which cause detrimental outages and overall degraded data rates. Motivated by this, in this work, we first derive the statistics of the incurred additive interference, and then analyze the rate of the considered NOMA-VLC channel. An analytical expression  is subsequently derived for the rate of NOMA-VLC systems with ambient light and user-mobility, followed by the formulation of a power-allocation technique for the underlying scenario, which has been shown to outperform classical gain-ratio power allocation in terms of achievable rate. The obtained analytical results are corroborated with computer
simulations for various realistic VLC scenarios of interest, which lead to useful insights of theoretical and practical interest. For example, it is shown that, in a NOMA-enabled VLC system, the maximum rate at which information can be transmitted over a static VLC communication channel with ambient light asymptotically converges to the Shannon Hartley capacity formula. 
\end{abstract}

%
\IEEEpeerreviewmaketitle
\section{Introduction}
Visible light communications (VLC) \cite{inan2009impact,le2010indoor,o2008visible} has emerged as a green, low-cost and secure communication paradigm to address some of the key challenges for the next-generation communication systems \cite{haas2016lifi}. In VLC systems, wireless information transmission is realized by modulating the intensity of the light emitting diode (LED) at a switching rate imperceptible to the human eye \cite{inan2009impact}, thereby achieving a dual goal of illumination and data-transmission. At the receiver terminal, a photodetector (PD) is employed to convert the fluctuations in the received beams into current that is used for data recovery.

Likewise, non-orthogonal multiple access (NOMA) has emerged as a viable approach to accommodate more users over the same time-frequency resources \cite{ding2016impact,ding2016general,ding2016application}. Users' signals can be superposition coded by assigning distinct power levels to different users according to their encountered channel-conditions. At the receivers, multi-user detection (MUD) and interference mitigation is realized by successive interference cancellation (SIC). In NOMA, SIC first decodes users with higher transmission power and then subtracts them from its received signals while treating all other users' signal as noise. It has been shown that power allocation significantly affects the performance of NOMA based VLC systems. Based on this, optimum and effective sub-optimum algorithms can result to higher data  rates as compared to orthogonal multiple access techniques e.g. orthogonal frequency division multiple access (OFDMA), which renders it attractive for VLC \cite{marshoud2016non,zhang2017user,lin2017experimental,rmitratran2017,marshoud2017performance}. However, the contributions in \cite{marshoud2016non,zhang2017user,lin2017experimental,rmitratran2017,marshoud2017performance} mostly model the overall additive noise-process as additive white Gaussian noise, and the corresponding NOMA based performance analysis is carried out under this assumption. {However, upon considering the effect of ambient light, the overall noise process for a NOMA-VLC system deviates from Gaussianity \cite{saleh2013photoelectron,pergoloni2015mutual}, which calls for a thorough investigation on the corresponding  impact in realistic communication scenarios. Based on this, the primary aim of the present contribution is to quantify the achievable capacity for realistic NOMA-VLC systems impaired by ambient light.}

Nevertheless, it is also recalled that despite the promising advantages of VLC, there are several performance-limiting impairments when considering user mobility, which has been mostly characterized by models such as the random walk model \cite{marshoud2016non}. In the same context, the authors in  \cite{yin2016performance} derived a statistical model for the probability density function (p.d.f) of channel-gain due to variation in location  of the users, which in turn resulted to an achievable rate reduction \cite{marshoud2017performance}. Another performance-impairing artefact that degrades the performance of VLC systems is the additive interference due to ambient light. The intensity of ambient light, which is modeled as polarized thermal light, can be statistically characterized by a chi-squared random variable \cite{saleh2013photoelectron,pergoloni2015mutual}, which results in additive interference and reduction of the overall SNR. Though there exists literature that indicates that ambient light can be suppressed by deployment of a blue filter \cite{islim2018impact}, its deleterious effects on the performance of VLC systems is non-negligible. Hence, analytical insights are essential for accurately determining the performance of such VLC systems impaired by ambient light and user-mobility.

The contributions of this work are summarized below:
\begin{itemize}
\item An exact analytical expression is derived for the p.d.f of the overall additive distortion in presence of ambient lighting and blue filter, which is an optical filter used for improving the overall signal to noise ratio in VLC systems.
\item The corresponding capacity for NOMA-VLC systems is quantified for \textit{static users}, and for scenarios with \textit{user-mobility} and additive interference due to \textit{ambient lighting}. 
\item Capitalizing on the above, and assuming negligible mobility, asymptotic equivalence of capacity of the considered channel to the Shannon-Hartley (SH) formula \cite{cover2012elements} is established for \textit{static} user scenarios.
\item A novel alternating projection based approach for effective power allocation is proposed for ambient light impaired NOMA scenarios, which is found to outperform classical gain-ratio power allocation (GRPA) in both static and mobility-impaired NOMA-VLC scenarios.
\item The derived analytical results are corroborated by extensive computer simulations, which verify their validity, and provide useful insights which will be useful in the design and deployment of VLC systems.  
\end{itemize}
The remainder of this paper is organized as follows: Section II describes the considered system model, and Section III characterizes the p.d.f for the additive distortion. Section IV, presents the simulation results along with useful discussions. Finally, conclusions are drawn in Section V.
\section{System Model}\label{Sec2}
In this section, we outline the system model adopted in this work. First, we describe the considered NOMA setup and then provide details on the considered noise model. 
\subsection{NOMA model}
In this subsection, the system model is presented for power domain NOMA-VLC.
Let $b^{(u)}$ denote the bitstream corresponding to the $u^{th}$ user, and $s^{(u)}$ denote modulation of $b^{(u)}$ to on-off keying (OOK). Also, let $s^{(u)}+\Gamma$ represent the transmitted signal, with $\Gamma$ denoting the DC offset, which is necessary in order to forward-bias the LED. 
To this effect, the received sample $x^{(u)}$ corresponding to transmitted symbol $s^{(u)}$ can be represented as
\begin{gather}\label{sysmod}
x^{(u)} = h^{(u)}\sum\displaylimits_{u=1}^{U}\sqrt{P^{(u)}}(s^{(u)}+\Gamma) + \underbrace{\alpha n^{(u)}+\beta \frac{(w^{(u)}-\nu)}{\sqrt{2\nu}}}_{\phi}
\end{gather}
where $h^{(u)}$ denotes the path gain of the corresponding $u^{th}$ user, $P^{(u)}$ denotes the power allocated to the $u^{th}$ user, $\Gamma$ denotes the DC offset, $n^{(u)}$ denotes additive white Gaussian noise of unit variance, which models thermal noise added at the photodetector, and $w^{(u)}$ denotes a Chi-squared distributed random variable with $\nu$ degrees of freedom, which models additive interference due to ambient light (see Section \ref{nm} for further details). {Notably, The user grouping is performed such as to maximize the sum-rate among those user partitions where the sum-rate are in ascending order \cite{8315498}\footnote{It is  noted that advanced methods of user grouping are beyond the scope of
the present contribution.}.}

The VLC channel model used in the considered set up is represented as follows \cite{marshoud2016non}:
\begin{equation}\label{pdfeqn1}
{h^{(u)}= \begin{cases} \frac{\mathcal{A}_{e}}{d^2\sin^{2}\Psi_{}}\mathcal{R}(\Theta_{})\cos\theta_{}, & \qquad {}0<\Theta<\Psi_{}\\
 0 &\qquad {}\text{otherwise}
 \end{cases}}
\end{equation}
where $\mathcal{A}_{e}$ denotes the area of photodetector, $d$ is the distance between the transmitter LED and the photodetector, $\Theta$ is the perpendicular angle of LED, and $\theta$ is the angle between transmit LED and photodetector with the receiver axis. Furthermore, 
$\Psi$ denotes the field of view (FOV) for the photodetector, $R(\Theta)$ denotes the Lambertian radiant intensity which can be written as $R(\Theta_{})=\frac{(m+1)\cos^{m}(\Theta)}{2\pi}$, where $m$ represents the order of Lambertian emission given as
$m = -\ln 2/\ln\cos\Theta_{\frac{1}{2}}$, with $\Theta_{\frac{1}{2}}$ denoting the semi-angle at half-power of LED.

It is worthwhile to mention that the users' channel-gains are sorted in ascending order, i.e., $h^{(1)}<h^{(2)}<h^{(3)}\cdots<h^{(U)}$, where $U$ is the number of users. 

For gain ratio power allocation \cite{marshoud2016non}, this implies that, $P^{(1)}>P^{(2)}\cdots>P^{(U)}$, and at each UE, users' symbols are recovered by an SIC. However, the monotonicity of power-coefficients is not generally guaranteed, and depends on chosen rate-quality of service (QoS) points \cite{vaezi2018non}. 
\subsection{Noise Model}
\label{nm}
This subsection details the considered statistical model for the additive interference.
To this end, $\alpha$ and $\beta$ denote variance/weighting factors, (in particular $\beta$ denotes the attenuation in ambient light after blue filtering, and $\alpha^2$ denotes the variance of the Gaussian component), $n^{(u)}$ is the additive white Gaussian noise at the $u^{th}$  UE with unit variance, and $w^{(u)}$ denotes the intensity of ambient light at $u^{th}$ UE, which can be modeled as a Chi-squared random variable with $\nu$ degrees of freedom (denoted by $\chi^2_{\nu}(\cdot)$) \cite{saleh2013photoelectron,pergoloni2015mutual}. {The additive term  $\phi$ denotes the sum of scaled Gaussian and central Chi-squared distributions. It can be noted that $\nu$ (the mean of Chi-squared r.v.) is subtracted from the $\chi^2_{\nu}(\cdot)$ random variable, since, a DC blocking capacitor is assumed to remove the DC offset $\Gamma$.} Additionally, the scaling factor $\sqrt{2\nu}$ makes the random variable unit variance for simulation of various signal-to-noise ratios. 
\subsection{Statistical modeling of user mobility}
To model user mobility, a statistical model for VLC channel can be found \cite{yin2016performance}, wherein the p.d.f of a line-of-sight VLC channel (denoted by $p(h)$) assuming user mobility is expressed as
\begin{equation}\label{pdfeqn1}
{p(h)= \begin{cases} K\text{ }h^{-\frac{2}{m+3}-1} &\qquad h \in [h_{\text{min}},h_{\text{max}}]\\
 0 &\qquad  \text{otherwise}
 \end{cases}}
\end{equation}
where 
\begin{gather}
h_{\text{min}}= \frac{K_{1}(m+1)L^{m+1}}{(r_{\text{max}}^2+L^2)^{\frac{m+3}{2}}}
\end{gather}
and 
\begin{gather}
h_{\text{max}}= \frac{K_{1}(m+1)L^{m+1}}{(r_{\text{max}})^{{m+3}}}
\end{gather}
with 
\begin{gather}
K_{1}=\frac{A_{p}}{2\pi}
\end{gather}
where $A_{p}$ is the effective photodetector area, $r_{\text{max}}$ is the maximum coverage radius, $L$ is the height of the LED, and  
\begin{gather}
\small
K=\frac{2K_{1}^{\frac{2}{m+3}}((m+1)L^{m+1})^{\frac{2}{m+3}}}{(m+3)r_{\text{max}}^2}.
\end{gather}
\normalsize
Notably, the CDF of $p(h)$ (denoted as $P(h)$) can be written as
\begin{equation}\label{pdfeqn1}
{P(h)= \begin{cases} \frac{h^{-\frac{2}{m+3}}-h_{\min}^{-\frac{2}{m+3}}}{-\frac{2}{(m+3)K}} & \qquad h \in [h_{\text{min}},h_{\text{max}}]\\
 0 &  \qquad h \in (-\infty,h_{\min}]\\
 1 &  \qquad h \in [h_{\max},\infty)
 \end{cases}}
\end{equation}
To this effect and dennoting constants 
\begin{gather}
f_{1}= \frac{-(m+3)K h_{\min}^{-\frac{2}{m+3}}}{2}
\end{gather} 
and 
\begin{gather}
f_{2} = \frac{-(m+3)K}{2}
\end{gather}
the p.d.f given for the $u^{th}$ layer of SIC, incorporating ordered channel statistics can be written as \cite{8501953}
\begin{equation}\label{pdfeqn1}
{p_{(u)}(h)= \begin{cases} \frac{U!K h^{-\frac{2}{m+3}-1}P(h)^{u-1}(1-P(h))^{L-u}}{(u-1)!(U-u)!} & \qquad h \in [h_{\text{min}},h_{\text{max}}]\\
 0 &\qquad  \text{otherwise}
 \end{cases}}
\end{equation}
Applying Binomial theorem, the above expression can be rewritten as
\begin{gather*}
p_{(u)}(h)=\frac{U!K}{(u-1)!(U-u)!}h^{-\frac{2}{m+3}-1}\times\\ \nonumber
\sum_{i=0}^{u-1}\sum_{j=0}^{U-u}{{u-1}\choose i}{{U-u}\choose j} (-f_{1})^{i}(f_{2}h^{-\frac{2}{m+3}})^{u-1-i}
(1+f_{1})^{j}(-f_{2}h^{-\frac{2}{m+3}})^{U-u-j}
\end{gather*}
which after some algebraic manipulations yields
\begin{gather}
p_{(u)}(h)=\frac{U!K}{(u-1)!(U-u)!}h^{-\frac{2}{m+3}-1}
\sum_{i=0}^{u-1}\sum_{j=0}^{U-u}{{u-1}\choose i}{{U-u}\choose j} (-f_{1})^{i}
(1+f_{1})^{j}(-f_{2}h^{-\frac{2}{m+3}})^{U-j-i-1}
\end{gather}
or
\begin{gather}\label{osc}
p_{(u)}(h)=\frac{U!K}{(u-1)!(U-u)!}
\sum_{i=0}^{u-1}\sum_{j=0}^{U-u}{{u-1}\choose i}{{U-u}\choose j} (-f_{1})^{i}
(-1)^{U-j-i-1}(1+f_{1})^{j}\times\\ \nonumber
(f_{2}^{U-j-i-1}h^{-\frac{2}{m+3}({U-j-i})-1}).
\end{gather}
It is noted that the derivation of this p.d.f using ordered statistics of the ordered channel gains, facilitates the  derivation of the corresponding  expressions for capacity in subsequent sections.  
\section{Performance Evaluation}
In this section, we analyze the considered NOMA setup with the new noise model, in order to quantify the corresponding  achievable rates in the presence of ambient light.
\subsection{Statistical analysis of noise p.d.f and capacity for static users}
In this subsection, we characterize the p.d.f of the additive distortion, and the major findings are summarized in the following two theorems. 
\begin{thm}	
\it{The additive distortion (denoted by $\phi$), which is given by the following equation: 
\begin{gather}
\label{phi_ref}
\phi=\alpha n^{(u)}+\beta \frac{(w^{(u)}-\nu)}{\sqrt{2\nu}} 
\end{gather}
is drawn from the following p.d.f.:
\begin{gather}
\label{p_phi_ref}
\allowdisplaybreaks
p(\phi)=\sum_{m=0}^{\infty}\frac{(\frac{\nu}{2})_{m}}{m!}\Big(\frac{2\beta}{\alpha}\Big)^{m}H_{m}\Big(\frac{(\phi+\beta^{'}\nu)}{\alpha}\Big)\exp\Big(-\frac{(\phi+\beta^{'}\nu)^2}{2\alpha^2}\Big)
\end{gather}
where $\beta^{'}=\frac{\beta}{\sqrt{2\nu}}$.
with $H_{m}(\cdot)$ denoting the $m^{th}$ Hermite polynomial, $(\cdot)_{m}$ denotes Pochhammer symbol, and $\Theta_{m}=\frac{(\frac{\nu}{2})_{m}}{m!}\Big(\frac{2\beta}{\alpha}\Big)^{m}$.
} 
\end{thm}
\begin{proof}
To prove Theorem 1, we first use a result in probability theory given in \cite{papoulis2002probability}, which states that the overall moment generating function (MGF) of a sum of independent scaled Gaussian and Chi-squared random variables can be expressed as \cite{papoulis2002probability}
\begin{gather}
\label{mgf}
M(t) = \frac{\exp(\frac{t^2\alpha^2}{2})}{(1-2\beta t)^{\frac{\nu}{2}}}.
\end{gather}
From the MGF, {expanding the denominator as a series, and using the inverse Laplace relation between MGF and p.d.f,} the corresponding p.d.f can be written as 
\begin{gather}
p(\psi) = \sum\displaylimits_{m=0}^{\infty}\frac{\Big(\frac{\nu}{2}\Big)_{m}(2\beta)^m}{m!}
\int\displaylimits_{-\infty}^{\infty}\exp(\frac{\alpha^2 t^2 \sigma^2}{2}) t^{m}\exp(-t\psi)d\psi.
\end{gather}
Noting that
\begin{gather}
\int\displaylimits_{-\infty}^{\infty}\exp(\frac{t^{2}\alpha^{2}}{2})\exp(-t\psi) d\psi= \frac{1}{\sqrt{2\pi\sigma^2}}\exp(-\frac{\psi^2}{2\alpha^2})
\end{gather}
and applying repeated derivatives with respect to $\psi$, we obtain the following \cite[eq. 8.959(1)]{gradshteyn2014table}
\begin{gather}
p(\psi) = \sum\displaylimits_{m=0}^{\infty}\frac{\Big(\frac{\nu}{2}\Big)_{m}}{m!}\Big(\frac{2\beta}{\alpha}\Big)^{m}H_{m}\left(\frac{\psi}{\alpha}\right)\exp(-\frac{\psi^2}{2\alpha^2}).
\end{gather}
Furthermore, since the weighting coefficients decrease monotonically, the expression for the p.d.f can be rewritten as
\begin{gather}
p(\psi)=\sum_{m=0}^{\infty}\frac{(\frac{\nu}{2})_{m}}{m!}\Big(\frac{2\beta}{\alpha}\Big)^{m}H_{m}\Big(\frac{\psi}{\alpha}\Big)\exp\Big(-\frac{\psi^2}{2\alpha^2}\Big)
\end{gather}
Noting the definition of $\phi$ as in (\ref{phi_ref}) using the scaling and shifting properties of p.d.f, and assuming $\beta^{'}=\frac{\beta}{\sqrt{2\nu}}$ the p.d.f of additive distortion, $\phi$, can be expressed as in (\ref{p_phi_ref}).
\end{proof}
\begin{corr}
Based on this, the CDF $P(\psi)$ can be written in the following explicit form:
\begin{gather}
\label{pl1}
P(\psi)
=\sum_{m=0}^{\infty}\frac{(\frac{\nu}{2})_{m}}{m!}\Big(\frac{2\beta}{\alpha}\Big)^{m}\Big[H_{m-1}(0)-\exp(-\frac{\psi^2}{2\alpha^2})H_{m-1}\Big(\frac{\psi}{\alpha}\Big)\Big].
\end{gather}
\end{corr}
\begin{proof}
Using (\ref{p_phi_ref}), the CDF $P(\psi)$ can be expressed in terms of the following series representation
\begin{gather}
P(\psi)=\sum_{m=0}^{\infty}\frac{(\frac{\nu}{2})_{m}}{m!}\Big(\frac{2\beta}{\alpha}\Big)^{m}\int_{-\infty}^{\psi}H_{m}\Big(\frac{\psi}{\alpha}\Big)\exp\Big(-\frac{\psi^2}{2\alpha^2}\Big)d\psi
\end{gather}
which with the aid of \cite[eq. (7.373)]{gradshteyn2014table}, can be explicitly given as per (\ref{pl1}).
\end{proof}
\begin{corr}
For the specific case of large enough $M$, the above p.d.f can be approximated as
\begin{gather}
\label{pa1}
p(\phi)=\hat{g}(\phi)\exp\Big(-\frac{(\phi+\beta^{'}\nu)^2}{2\alpha^2}\Big)
\end{gather}
where
\begin{gather}
\label{pa2}
\hat{g}(\phi)=\sum_{m=1}^{M}\Theta_{m}H_{m}\Big(\frac{\phi+\beta^{'}\nu}{\alpha}\Big)
\end{gather}
\end{corr}
\begin{proof}
For large $M$ {(i.e. number of terms in series)}, the CDF of $\psi$ can be approximated as
\begin{gather}
P(\psi)\approx\sum_{m=0}^{M}\frac{(\frac{\nu}{2})_{m}}{m!}\Big(\frac{2\beta}{\alpha}\Big)^{m}\Big[H_{m-1}(0)-\exp(-\frac{\psi^2}{2\alpha^2})H_{m-1}\Big(\frac{\psi}{\alpha}\Big)\Big].
\end{gather}
Also, an approximate p.d.f can be derived by taking the  derivative of $P(\psi)$, and using Rodigues formula \cite[eq. 8.959(1)]{gradshteyn2014table}, which yields
\begin{gather}
p(\psi)=\frac{dP(\psi)}{d\psi}=\hat{g}(\psi)\exp\Big(-\frac{\psi^2}{2\alpha^2}\Big)
\end{gather}
where
\begin{gather}
\hat{g}(\psi)=\sum\limits_{m=1}^{M}\Theta_{m}H_{m}\Big(\frac{\psi}{\alpha}\Big)
\end{gather}
and 
\begin{gather}
\Theta_{m}=\frac{(\frac{\nu}{2})_{m}}{m!}\Big(\frac{2\beta}{\alpha}\Big)^{m}
\end{gather}
for large $M$. Noting the definition of $\phi$ as in (\ref{phi_ref}) using the scaling and shifting properties of p.d.f, and assuming $\beta^{'}=\frac{\beta}{\sqrt{2\nu}}$, we can arrive at (\ref{pa1}) and (\ref{pa2}).
\end{proof}
\begin{corr}
For high $\frac{\nu}{2}$, it follows that
\begin{gather}
 \label{final1}
p(\phi) = \frac{1}{Z}\exp(\frac{2\nu\beta^{'}(\phi+\beta^{'}\nu)-(\phi+\beta^{'}\nu)^2}{2\alpha^2}) 
 \end{gather}
 where $Z$ is a suitable partition function.
\end{corr}
\begin{proof}
For large values of $\nu$ and $\beta<\alpha$ (i.e. assuming that the power of Chi-squared random variable is lower than the Gaussian random variable after blue filtering), we note that $(\frac{\nu}{2})_{m}\approx(\frac{\nu}{2})^{m}$. Based on this we can further approximate $p(\psi)$ as \cite[eq. 8.957(1)]{gradshteyn2014table}
 \begin{gather}
 \label{pdf_psi}
p(\psi) = \frac{1}{Z}\exp(\frac{2\nu\beta\psi-\psi^2}{2\alpha^2}) 
 \end{gather}
where 
\begin{gather}
Z=\int\displaylimits_{-\infty}^{\infty}\exp(\frac{2\nu\beta\psi-\psi^2}{2\alpha^2})d\psi=\exp(\frac{\nu^{2}\beta^{2}}{2\alpha^2})\sqrt{2\pi\alpha^2} 
\end{gather}
is the partition function.
Hence, using the scaling and shifting properties of MGF, and assuming $\beta^{'}=\frac{\beta}{\sqrt{2\nu}}$ the p.d.f of additive distortion can be expressed as
 \begin{gather}
 \label{final_11}
p(\phi) = \frac{1}{Z}\exp(\frac{2\nu\beta^{'}(\phi+\beta^{'}\nu)-(\phi+\beta^{'}\nu)^2}{2\alpha^2}). 
 \end{gather}
 which proves the result.
 \end{proof}
\begin{remark}
One can express the exact expression for the p.d.f as
\begin{gather}
p(\psi) = \sum\limits_{m=0}^{\infty}\frac{\Big(\frac{\nu}{2}\Big)_{m}}{m!} \Big(\frac{2\beta}{\alpha}\Big)^{m} H_{m}\Big(\frac{\psi}{\alpha}\Big)
\end{gather}
which can be equivalently re-written as an inverse Wishart transform, namely,
\begin{gather}
p(\psi) = \exp(-\frac{D^2}{2})\sum\limits_{m=0}^{\infty}\frac{\Big(\frac{\nu}{2}\Big)_{m}}{m!} \Big(\frac{2\beta\psi}{\alpha^2}\Big)^{m} 
\end{gather}
where $D$ denotes the differentiation operator w.r.t $\psi$. After some algebraic manipulations, the above equation can be rewritten as
\begin{gather}
p(\psi) = \exp(-\frac{D^2}{2})\Bigg[\frac{1}{\Big(1-\frac{2\beta\psi}{\alpha^2}\Big)^{\frac{\nu}{2}}}\Bigg]
\end{gather}
where
\begin{gather}
D^{2m}\Bigg[\frac{1}{(1-\frac{2\beta\psi}{\alpha})^{\frac{\nu}{2}}}\Bigg]=\frac{(\frac{\nu}{2})_{2m}(\frac{2\beta}{\alpha})^{2m}}{\Big(1-\frac{2\beta\psi}{\alpha^2}\Big)^{\frac{\nu}{2}+2m}}
\end{gather}
Based on the above, it follows that,
\begin{gather}
p(\psi) = \frac{\sum\limits_{m=0}^{\infty}\frac{(\frac{nu}{4})_{m}(\frac{2+\nu}{4})_{m}}{m!} \Big[\frac{-8\beta^2}{\alpha^{2}(1-\frac{2\beta\psi}{\alpha^2})^2}\Big]^m}{\Big(1-\frac{2\beta\psi}{\alpha}\Big)^{\frac{\nu}{2}}}\exp(-\frac{\psi^2}{2\alpha^2})
\end{gather}
which can be expressed in closed-form in terms of the confluent hypergeometric function, namely
\begin{gather}
\label{ffi}
p(\psi) = \frac{\pFq{2}{0}{\frac{\nu}{4},\frac{2+\nu}{4}}{}{\frac{-8\beta^2}{\alpha^{2}(1-\frac{2\beta\psi}{\alpha^2})^2}}}{\Big(1-\frac{2\beta\psi}{\alpha}\Big)^{\frac{\nu}{2}}}\exp(-\frac{\psi^2}{2\alpha^2}).
\end{gather}
It is evident that assuming $\beta<<\alpha$, the above equation can be accurately approximated by (\ref{final_11}). 
Another form of (\ref{ffi}) can be deduced by invoking the property of hypergeometric function in \cite[eq. (07.31.03.0083.01)]{bworld}, which yields
\begin{gather}
p(\psi)=\alpha^{\frac{\nu}{2}}\beta^{-\frac{\nu}{2}}H_{-\frac{\nu}{2}}\Big(\frac{\alpha(1-\frac{2\beta\psi}{\alpha^2})}{2\sqrt{2}\beta}\Big)\exp(-\frac{\psi^2}{2\alpha^2}).
\end{gather}
Furthermore, from \cite[eq. (07.01.02.0001.01)]{bworld}, the above equation can be further simplified as
\begin{gather}
p(\psi)=\alpha^{\frac{\nu}{2}}\beta^{-\frac{\nu}{2}}2^{\frac{\nu}{2}}\sqrt{\pi}\Bigg[\frac{\pFq{1}{1}{\frac{\nu}{4}}{\frac{1}{2}}{\frac{\alpha^{2}(1-\frac{2\beta\psi}{\alpha^2})^2}{8\beta^2}}}{\Gamma[\frac{k+2}{4}]}-\frac{\frac{\alpha}{\sqrt{2}\beta}(1-\frac{2\beta\psi}{\alpha^2})}{\Gamma[\frac{k}{4}]}\pFq{1}{1}{\frac{2+\nu}{4}}{\frac{3}{2}}{\frac{\alpha^{2}(1-\frac{2\beta\psi}{\alpha^2})^2}{8\beta^2}}
\Bigg]\exp(-\frac{\psi^2}{2\alpha^2})
\end{gather}
\end{remark}
\begin{remark}
The series diverges for $\beta>\alpha$, which is the reason VLC is not viable when ambient light dominates severely, and justifies the failure of VLC when ambient light is predominant.
\end{remark}
\begin{thm}\label{app1}
\it{With blue filtering, the formula for rate in the considered scenario for small $\beta<<\alpha$ is asymptotically identical to the SH formula, which is denoted here by $R_{sh}^{(u)}$ for the $u^{th}$ user, and can be expressed as
\begin{gather}
\label{ab1}
R_{sh}^{(u)}=\frac{1}{2}\log(1+\frac{P^{(u)}}{\sum\limits_{l>u}P^{(l)}+
	\frac{\alpha^2}{|h^{(u)}|^2}})
\end{gather}
}
\end{thm}
\begin{proof}
We divide our proof into two parts as given below. 
\subsubsection{Expression for Entropy of sum of Gaussian and Chi-squared random variable}
From (\ref{final_11}), we can conclude that $\mathbb{E}[\phi] = 0$, and $\text{var}[\phi]=\alpha^2+\beta^2$. Taking the $\log$ of $p(\phi)$, and the statistical expectation on both sides yields
 \smallskip
 \begin{gather}
 \label{final}
-\mathbb{E}[\log[p(\phi)]] = \mathbb{E}[\log(Z)] - \newline\frac{\mathbb{E}[2\nu\beta^{'}(\phi+\beta^{'}\nu)-(\phi^2+\nu^2\beta^{'^{2}}+2\phi\nu\beta^{'})]}{2\alpha^2}
 \end{gather}
 \normalsize
To this effect, and recalling that $\mathbb{E}[\phi] = 0$, and $\text{var}[\phi]=\alpha^2+\beta^2$, the above expression can be further simplified as follows:
\begin{gather}
-\mathbb{E}[\log[p(\phi)]] = \mathbb{E}[\log(Z)] - \newline\frac{\mathbb{E}[2\nu^2\beta^{'^{2}}-\alpha^2-\beta^2+2\phi\nu\beta^{'}-\nu^2\beta^{'^2}]}{2\alpha^2}
\end{gather}
Using expression for the partition function $Z$, we derive
\begin{gather}
\log Z = \frac{\nu^2\beta^{'^2}}{2\alpha^2}+\frac{1}{2}\log(2\pi\alpha^2)
\end{gather}
and hence (\ref{final}) can be written as
\begin{gather}
\label{final22}
-\mathbb{E}[\log[p(\phi)]] =  \Big(\frac{\alpha^2+\beta^2}{2\alpha^2}\Big)+\frac{1}{2}\log(2\pi\alpha^2).
\end{gather}
\subsubsection{Capacity}
 In this subsection, we derive a formula for the rate of the considered NOMA-VLC channel. First, we note that from the principle of maximum entropy distributions, and inspection of (\ref{pdf_psi}), 
as well as noting only first and second order terms in the exponent, one can conclude that for a given mean $\mathbb{E}[\phi]=0$, and $\mathbb{E}[\phi^2]=\alpha^2+\beta^2$, $p(\phi)$ is an entropy maximizing distribution \cite{cover2012elements}. To this effect, considering two random variables $\mathcal{Y}$ and $\mathcal{Z}$ in the considered setup, which are related by the following equation 
 \begin{gather}
 \mathcal{Y}=\mathcal{Z}+\mathcal{P}_{1}+\mathcal{P}_{2}
 \end{gather} 
 where $\mathcal{Z}$ is a Gaussian term consisting of the desired user's signal of power $P^{(u)}$, $\mathcal{P}_{1}$ is the sum of a suitable equivalent Gaussian ($\mathcal{P}_{1}$ that accounts for the superposition of signals $l>u$ and AWGN) and $\mathcal{P}_{2}$ is a centered normalized Chi-squared random variable ($\mathcal{P}_{2}$).
 Based on this, the expression for the rate can be written as
 \begin{gather}
 \label{pas1}
 R = \max_{\mathbb{E}[|\mathcal{Z}|^2]\leq P_{tot},\text{  }\mathbb{E}[\mathcal{Y}]=0} H(\mathcal{Y}) - H(\mathcal{Y}|\mathcal{Z})
 \end{gather}
 where $H(\cdot)$ denotes the corresponding Shannon entropy, and $P_{tot}$ denotes the total power budget. Hence, (\ref{pas1}) can be equivalently expressed as 
 \begin{gather}
R = \max_{\mathbb{E}[|\mathcal{Z}|^2]\leq P_{tot},\text{  }\mathbb{E}[\mathcal{Y}]=0} H(\mathcal{Y}) - H(\mathcal{Z}+\mathcal{P}_{1}+\mathcal{P}_{2}|\mathcal{Z})
\end{gather}
and
\begin{gather}
R = \max_{\mathbb{E}[|\mathcal{Z}|^2]\leq P_{tot},\text{  }\mathbb{E}[\mathcal{Y}]=0} H(\mathcal{Y}) - H(\mathcal{P}_{1}+\mathcal{P}_{2}|\mathcal{Z}).
\end{gather}
Based on this and since $\mathcal{P}$ and $\mathcal{Z}$ are statistically independent, it follows that
\begin{gather}
R = \max_{\mathbb{E}[|\mathcal{Z}|^2]\leq P_{tot},\text{  }\mathbb{E}[\mathcal{Y}]=0} H(\mathcal{Y}) - H(\mathcal{P}_{1}+\mathcal{P}_{2}).
\end{gather}
{Applying the central limit theorem \footnote{{It is worth noting that assuming a large number of users, the distribution \textit{tends} to a Gaussian. In case of less users, the Gaussian assumption provides a lower-bound on the rate (since given a covariance matrix, a Gaussian random variable maximizes differential entropy) which becomes tighter with an increase in number of users.}}, and assuming $\mathcal{Y}$ to be sum of Gaussian random variables with power $P^{(u)}$ (for the desired user) added with interference terms $(l>u)$ from the subsequent layers of SIC), and Chi-squared distribution, we can write $H(\mathcal{Y})$ as}
\begin{gather}
\label{s1}
H(\mathcal{Y}) = \frac{1}{2}\log(2\pi\Bigg[P^{(u)}+\sum\limits_{l>u}P^{(l)}+
\frac{\alpha^2}{|h^{(u)}|^2}\Bigg])+
\frac{1}{2}+\frac{\beta^2}{P^{(u)}+\sum\limits_{l>u}P^{(l)}+\frac{\alpha^2}{|h^{(u)}|^2}}
\end{gather}
Similarly,
\begin{gather}
\label{s2}
H(\mathcal{P}) = \frac{1}{2}\log(2\pi\Bigg[\sum\limits_{l>u}P^{(l)}+
\frac{\alpha^2}{|h^{(u)}|^2}\Bigg])+\frac{1}{2}+\frac{\beta^2}{\sum\limits_{l>u}P^{(l)}+
	\frac{\alpha^2}{|h^{(u)}|^2}}
\end{gather}
Therefore
\begin{gather}
\label{final_rate}
R^{(u)}= \frac{1}{2}\log(1+\frac{P^{(u)}}{\sum\limits_{l>u}P^{(l)}+
	\frac{\alpha^2}{|h^{(u)}|^2}})-\frac{\beta^2 P^{(u)}}{\Big(\sum\limits_{l>u}P^{(l)}+
	\frac{\alpha^2}{|h^{(u)}|^2}\Big)\Big(P^{(u)}+\sum\limits_{l>u}P^{(l)}+
	\frac{\alpha^2}{|h^{(u)}|^2}\Big)}
\end{gather}
\normalsize
{This indicates that the rate of each user is lowered due to the influence of ambient light. However, if $\beta<<\alpha$ after the application of a blue-filter, it can be inferred that $R^{(u)}\approx R_{{sh}}^{(u)}$ (where $R_{{sh}}^{(u)}$ is defined in (\ref{ab1})), which proves the result.} 
\end{proof}
\subsection{Capacity for scenario with user mobility}
\label{rateq}
In this section, we quantify the rate for each user under the assumption of user mobility. For quantification of capacity in the presence of user mobility, we seek to take the expected value of (\ref{final_rate}) with respect to $p(h^{(u)})$. The expected capacity, say $\rho^{(u)}$ is expressed as
\begin{gather}
\rho^{(u)} = \mathbb{E}_{h}[R^{(u)}].
\end{gather}
We also note that (\ref{final_rate}) can be expanded as
\begin{gather}
\label{ag}
R^{(u)}=\frac{1}{2}\log(|h^{(u)}|^2(\sum\limits_{l\geq u}P^{(l)})+\alpha^2)-\frac{1}{2}\log(|h^{(u)}|^2(\sum\limits_{l>u}P^{(l)})+\alpha^2)\\ \nonumber
-\beta^2\Bigg[\frac{|h^{(u)}|^2}{\sum\limits_{l>u}|h^{(u)}|^2P^{(l)}+\alpha^2}-\frac{|h^{(u)}|^2}{\sum\limits_{l\geq u}|h^{(u)}|^2P^{(l)}+\alpha^2}\Bigg].
\end{gather}
We first prove the following two lemmas for carrying out the analysis
\begin{lemma}\label{app3}
\begin{gather}
\label{lem1}
\mathcal{I}(x,\theta,b,c)=\int x^{-\theta} \log\small(bx+c) dx=\frac{x^{1-\theta}(\pFq{2}{1}{1,1-\theta}{2-\theta}{-\frac{bx}{c}}-(\theta-1)\log\small(bx+c)-1)}{(\theta-1)^2}
\end{gather}
\end{lemma}
\begin{proof}
Integrating by parts, yields
\begin{gather}
\mathcal{I}(x,\theta,b,c)=\log(bx+c)\frac {x^{-\theta+1}}{-\theta+1} - \int \frac{b}{bx+c}\frac{x^{1-\theta}}{-\theta+1} dx
\end{gather}
which can be alternatively expressed as 
\begin{gather}
\mathcal{I}(x,\theta,b,c)=\log(bx+c)\frac {x^{-\theta+1}}{-\theta+1} - \frac{x^{-\theta+1}}{(-\theta+1)^2}\int \ {x^{-\theta}}\Big(1+\frac{bx}{c}\Big)^{-1} dx
\end{gather}
or
\begin{gather}
\mathcal{I}(x,\theta,b,c)=\log(bx+c)\frac {x^{-\theta+1}}{-\theta+1} - \frac{x^{-\theta+1}}{(-\theta+1)^2}+\frac{1}{1-\theta}\int \ {x^{-\theta}}\Big(1+\frac{bx}{c}\Big)^{-1} dx
\end{gather}
and
\begin{gather}
\mathcal{I}(x,\theta,b,c)=\log(bx+c)\frac {x^{-\theta+1}}{-\theta+1} - \frac{x^{-\theta+1}}{(-\theta+1)^2}+\frac{1}{1-\theta}\int \ {x^{-\theta}}\pFq{1}{0}{1}{}{-\frac{bx}{c}} dx
\end{gather}
where $\pFq{p}{q}{}{}{}$ denotes the hypergeometric function.
Based on the above and with the aid of \cite[eq. (07.19.21.0002.01)]{bworld} we can infer
\begin{gather}
\mathcal{I}(x,\theta,b,c)=\log(bx+c)\frac {x^{-\theta+1}}{-\theta+1} - \frac{x^{-\theta+1}}{(-\theta+1)^2}+\frac{x^{-\theta+1}}{(1-\theta)^2} \pFq{2}{1}{1-\theta,1}{2-\theta}{-\frac{bx}{c}} dx
\end{gather}
which proves (\ref{lem1}), and completes the proof.
\end{proof}
\begin{lemma}\label{app4}
\begin{gather}
\mathcal{B}(x,\theta,k1,k2,\zeta) =\int \frac{\beta^2 x^{-\theta+1}}{\zeta^2} \Bigg[\pFq{1}{0}{1}{}{-\frac{k_{1}x}{\zeta^2}}
-\pFq{1}{0}{1}{}{-\frac{k_{2}x}{\zeta^2}}\Bigg]dx\\\nonumber
=\frac{\beta^{2}x^{-\theta+2}}{\zeta^{2}({-\theta+2})}\Bigg[\pFq{2}{1}{1,2-\theta}{3-\theta}{-\frac{k_{1}x}{\zeta^2}}-\pFq{2}{1}{1,2-\theta}{3-\theta}{-\frac{k_{2}x}{\zeta^2}}\Bigg]
\end{gather}
\end{lemma}
\begin{proof}
The result follows directly from \cite[eq. (07.19.21.0002.01)]{bworld}.
\end{proof}
Using these lemmas and expression for the ordered p.d.f given in (\ref{osc}), we can arrive at the following proposition.
\begin{prop}
Using Lemma \ref{app3} and Lemma \ref{app4}, and (\ref{ag}) one can find an expression for the mean rate for a given $\theta$ as $\rho^{(u)}_{1}(\theta)$ as follows
\begin{gather}
\label{srmobile}
\rho^{(u)}_{1}(\theta)=\int\limits_{h{\min}}^{h_{\max}}h^{(u)^{-\theta}}R^{(u)}dh^{(u)} \\ \nonumber
=\mathcal{I}(h_{\max}^2,\theta,\sum\limits_{l\geq u}P^{(l)},\alpha^2)-
\mathcal{I}(h_{\max}^2,\theta,\sum\limits_{l> u}P^{(l)},\alpha^2)
+\mathcal{B}(h_{\max}^2,\theta,\sum_{l> u}P^{(l)},\sum_{l \geq u}P^{(l)},\alpha^2)\\\nonumber-\mathcal{B}(h_{\min}^2,\theta,\sum_{l> u}P^{(l)},\sum_{l \geq u}P^{(l)},\alpha^2)
-[\mathcal{I}(h_{\min}^2,\theta,\sum\limits_{l\geq u}P^{(l)},\alpha^2)-
\mathcal{I}(h_{\min}^2,\theta,\sum\limits_{l> u}P^{(l)},\alpha^2)]
\end{gather}
Using this expression, along with (\ref{osc}), we have the following expression for the rate for the $u^{th}$ user 
\begin{gather}
\rho^{(u)}=\frac{U!K}{(u-1)!(U-u)!}\\\nonumber
\sum_{i=0}^{u-1}\sum_{j=0}^{U-u}{{u-1}\choose i}{{U-u}\choose j} (-f_{1})^{i}(-1)^{U-j-i-1}
\Big(1+f_{1}\Big)^{j}\Big(f_{2}^{U-j-i-1}\rho^{(u)}_{1}\Big(\frac{1}{m+3}({U-j-i})+1\Big)\Big).
\end{gather}
Hence, the sum rate for the scenario with user mobility and ambient light interference can be given as
\begin{gather}\label{sr_mob}
\sum_{u=1}^{U}\rho^{(u)}.
\end{gather}
\end{prop}
\subsection{Power allocation for static users}
In this section, we derive a technique for power allocation based on the derived capacity. First, it is assumed that each user has a QoS requirement $R_{th}^{(u)}$, and there is an overall total power budget of $P_{tot}$. Let us further denote the rate-QoS constraint without 
ambient light, $\mathcal{R}^{(u)}$, as
\begin{gather}
\mathcal{R}^{(u)}=\frac{1}{2}\log(1+\frac{P^{(u)}}{\sum\limits_{l>u}P^{(l)}+
	\frac{\alpha^2}{|h^{(u)}|^2}})
\end{gather}
Under these conditions, we seek to optimize the following
\begin{equation*}
\begin{aligned}
& \underset{\textbf{p}}{\text{maximize}_{}}
& & \sum_{\forall u} R^{(u)}\\
& \text{subject to}
& & \mathcal{R}^{(u)} \geq R_{th}^{(u)}, \; i = 1, \ldots, m.\\
& & &\textbf{1}^{T}\textbf{p}=P_{tot}.
\end{aligned}
\end{equation*}
where $\textbf{p}$ denotes the vector of powers $P^{(u)}$ of each user, and $R^{(u)}$ is given as
\begin{gather}
\label{rran}
R^{(u)}= \frac{1}{2}\log(1+\frac{P^{(u)}}{\sum\limits_{l>u}P^{(l)}+
	\frac{\alpha^2}{|h^{(u)}|^2}})-\frac{\beta^2 P^{(u)}}{\Big(\sum\limits_{l>u}P^{(l)}+
	\frac{\alpha^2}{|h^{(u)}|^2}\Big)\Big(P^{(u)}+\sum\limits_{l>u}P^{(l)}+
	\frac{\alpha^2}{|h^{(u)}|^2}\Big)}
\end{gather}
It is worthwhile to note that the rate-QoS constraint is mentioned in this formulation without taking ambient light into consideration, since it is assumed that the ``naive" system designer would rely on the fact that most ambient light is supressed by a blue filter. Further, since $R^{(u)}-\mathcal{R}^{(u)}$ is bounded, the maximum difference is taken care by introduction of slack variables in the optimization problem (discussed below). 

Further, assuming $P^{(u)}>>\Big(\sum\limits_{l>u}P^{(l)}+
	\frac{\alpha^2}{|h^{(u)}|^2}\Big)$, $R^{(u)}$ can be approximated as
\begin{gather}
\label{rateapprox}
R^{(u)}\approx \frac{1}{2}\log(1+\frac{P^{(u)}}{\sum\limits_{l>u}P^{(l)}+
	\frac{\alpha^2}{|h^{(u)}|^2}})-\frac{\beta^2 }{\Big(\sum\limits_{l>u}P^{(l)}+
	\frac{\alpha^2}{|h^{(u)}|^2}\Big)}
\end{gather}
Let us denote the interference and noise power for each user as $$I^{(u)}=\Big(\sum\limits_{l>u}P^{(l)}+
	\frac{\alpha^2}{|h^{(u)}|^2}\Big)$$
Further, taking the derivative with respect to $P^{(u)}$, we get,
\begin{gather}
\frac{\partial \sum\limits_{\forall u}R^{(u)}}{{\partial P^{(u)}}}=\frac{1}{2\Big(P^{(u)}+\sum\limits_{l>u}P^{(l)}+\frac{\alpha^{2}}{|h^{(u)}|^2}\Big)}-\sum\limits_{\forall q<u}\Bigg[\frac{1}{2\Big(\sum\limits_{l>q}P^{(l)}+\frac{\alpha^2}{|h^{(q)|^2}}\Big)}-\frac{\beta^2}{\Big(\sum\limits_{l>q}P^{(l)}+\frac{\alpha^2}{|h^{(q)}|^2}\Big)^2}\Bigg]
\end{gather}
or
\begin{gather}
\frac{\partial \sum\limits_{\forall u}R^{(u)}}{{\partial P^{(u)}}}=\frac{1}{2\Big(P^{(u)}+I^{(u)}\Big)}-\sum\limits_{\forall q<u}\Bigg[\frac{1}{2 I^{(q)}}-\frac{\beta^2}{I^{(q)^2}}\Bigg]
\end{gather}
Equating the derivative to zero, we get
\begin{gather}
\label{grad}
\frac{P^{(u)}}{I^{(u)}}+1=\frac{\frac{2}{I^{(u)}}}{\sum\limits_{\forall q<u} \Big[\frac{1}{I^{(q)}}-\frac{\beta^2}{I^{(q)^2}}\Big]}
\end{gather}
Next, for simplicity, we solve the equivalent relaxed equality constrained problem given as follows
\begin{equation}
\label{cop}
\begin{aligned}
& \underset{\textbf{p}}{\text{maximize}_{}}
& & \sum_{\forall u} R^{(u)}\\
& \text{subject to}
& &\mathcal{R}^{(u)} = R_{th}^{(u)}+\eta^{(u)}, \; i = 1, \ldots, m.\\
& & &\textbf{1}^{T}\textbf{p}=P_{tot}.\\
& & & \eta^{(u)}\geq 0
\end{aligned}
\end{equation}
where $\eta^{(u)}$ denotes a slack variable, and $R_{th}^{(u)}$ denotes the rate-QoS threshold.
Let us denote $$\kappa^{(u)} = 2^{R_{th}^{(u)}+\eta^{(u)}}-1 = \frac{P^{(u)}}{I^{(u)}}$$
Thus, we can rewrite, (\ref{grad}) as
\begin{gather*}
2^{R^{th}+\eta^{(u)}}=\frac{\frac{(2^{R_{th}^{(u)}+\eta^{(u)}}-1)}{2P^{(u)}}}{{\sum\limits_{\forall q<u}\Big[\frac{(2^{R_{th}^{(q)}+\eta^{(q)}}-1)}{2P^{(q)}}-\frac{\beta^2({2^{(R^{th}+\eta^{(q)})}-1)^2}}{P^{(q)^2}}}\Big]}
\end{gather*}
Rearranging, we get
\begin{gather}
\label{pafinal}
P^{(u)}=\frac{\frac{(2^{R_{th}^{(u)}+\eta^{(u)}}-1)}{2^{R_{th}^{(u)}+\eta^{(u)}+1}}}{{\sum\limits_{\forall q<u}\Big[\frac{2^{R_{th}^{(q)}+\eta^{(q)}}-1}{2 P^{(q)}}-\frac{\beta^2(2^{R_{th}^{(q)}+\eta^{(q)}}-1)^2}{P^{(q)^2}}}\Big]}
\end{gather}
This equation provides a recursion for allocating power to all the users. First the users are ordered in ascending power of channel-gains. Next, (\ref{pafinal}) is used for allocating power to users with a given rate based QoS.

Next, power allocation is performed assuming $\eta^{(u)}=0$. It will be trivial to note that for $\eta^{(u)}>0$.
\begin{gather}
\label{pafinal1}
P^{(u)}=\frac{\frac{(2^{R_{th}^{(u)}}-1)}{2^{R_{th}^{(u)}+1}}}{{\sum\limits_{\forall q<u}\Big[\frac{2^{R_{th}^{(q)}}-1}{2P^{(q)}}-\frac{\beta^2(2^{R_{th}^{(q)}}-1)^2}{P^{(q)^2}}}\Big]}+\Omega^{(u)}
\end{gather}
for some variable $\Omega^{(u)}$.
Next we project $\textbf{p}$ on the hyperplane $\textbf{1}^{T}\textbf{p}=P_{tot}$ as
\begin{gather} 
\textbf{p}^{\text{proj}} = \textbf{p}-(\textbf{1}^{T}\textbf{p}-P_{tot})\textbf{1}/U
\end{gather}
where $U$ denotes number of users and $\textbf{p}^{\text{proj}}$ denotes the power projected on the hyperplane $\textbf{1}^{T}\textbf{p}=P_{tot}$.

For the next iteration, $\Omega^{(u)}$ can be found by considering 
\begin{gather} 
\label{gafinal}
\textbf{p}^{\text{proj}} = (\textbf{p}+\mathbf{\Omega})-\Big(\textbf{1}^{T}(\textbf{p}+\mathbf{\Omega})-P_{tot}\Big)\frac{\textbf{1}}{U}
\end{gather}
or
\begin{gather}
\label{gafinal1}
\mathbf{\Omega}=\Big(\textbf{I}_{U}-\frac{\textbf{1}\textbf{1}^{T}}{U}\Big)^{-1}\Big[\textbf{p}^{\text{proj}}-\Big(\textbf{I}_{U}-\frac{\textbf{1}\textbf{1}^{T}}{U}\Big)\textbf{p}-P_{tot}\frac{\textbf{1}}{U}\Big]
\end{gather}
Using (\ref{pafinal}), (\ref{pafinal1}), (\ref{gafinal}), (\ref{gafinal1}), the proposed alternating projections based power-allocation technique is formulated in Algorithm \ref{Alg1}. 
\begin{algorithm}
\caption{Determine power allocation coefficients}
\begin{algorithmic}
\label{Alg1} 
\STATE \textbf{Initialize} $\mathbf{\Omega}=\mathbf{0}$, $\textbf{p}=\frac{\textbf{1}}{U}$, $\epsilon$.
\WHILE{$\|\textbf{p}-\textbf{p}^{\text{prev}}\|>\epsilon$}
\STATE $\textbf{p}^{\text{prev}}=\textbf{p}$.
\STATE Calculate $\textbf{p}$ using (\ref{pafinal1}).  
\STATE Calculate $\textbf{p}^{\text{proj}}$ using (\ref{gafinal}).
\STATE Calculate $\mathbf{\Omega}$ using (\ref{gafinal1}).
\STATE Set $\textbf{p}:=[\textbf{p}^{\text{proj}}]_{+}$, where $[\cdot]_{+}=(\cdot)\times((\cdot)>0)$.
\ENDWHILE
\end{algorithmic}
\end{algorithm}
Further, it can be noted that even though the original objective functions is non-convex, we find optima by alternating projections over a convex subset (simplices) defined by the equality constraints.
\subsection{Power allocation for channel with user mobility}
In this section, we describe power allocation for the scenario in presence of user mobility. From Section. \ref{rateq}, and the previous sections, we can identify that the original rate equations are highly  involved, and hence direct maximization of the sum-rate is complicated.

Hence, for simplicity, we consider the rate equation (\ref{rateapprox}), and replace it by an approximate rate,
$\hat{R}^{(u)}$, using Jensen's inequality such as (actually it is an upper bound on the rate)
\begin{gather} 
\hat{R}^{(u)}\approx \frac{1}{2}\log(1+\frac{P^{(u)}}{\sum\limits_{l>u}P^{(l)}+\frac{\alpha^2}{\mu_{h^2}}})-\frac{\beta^2}{\sum\limits_{l>u}P^{(l)}+\frac{\alpha^2}{\mu_{h^2}^2}} 
\end{gather}
where $\mu_{h^2}=\mathbb{E}[h^{{(u)}^2}]$. 
Next, we approximate the original rate constraint as,
$$\mathbb{E}\Big[\log_{2}\Big(1+\frac{P^{(u)}}{I^{(u)}}\Big)\Big]\approx
\mathbb{E}\Big[\log_{2}\Big(1+\frac{P^{(u)}}{\hat{I}^{(u)}}\Big)\Big]=R_{th}^{(u)}+\eta^{(u)}$$ wherein $$\hat{I}^{(u)}=\sum\limits_{l>u}P^{(l)}+\frac{\alpha^2}{\mu_{h^2}}$$
Using these approximations, we arrive at the following similar optimization problem
\begin{equation}
\label{copt1}
\begin{aligned}
& \underset{\textbf{p}}{\text{maximize}_{}}
& & \sum_{\forall u} \hat{R}^{(u)}\\
& \text{subject to}
& &\hat{\mathcal{R}}^{(u)} = R_{th}^{(u)}+\eta^{(u)}, \; i = 1, \ldots, m.\\
& & &\textbf{1}^{T}\textbf{p}=P_{tot}.\\
& & & \eta^{(u)}\geq 0
\end{aligned}
\end{equation}
where $\hat{\mathcal{R}}^{(u)} = \frac{1}{2}\log(1+\frac{P^{(u)}}{\sum\limits_{l>u}P^{(l)}+\frac{\alpha^2}{\mu_{h^2}}})$. Solving this problem, it is not difficult to conclude that we arrive at (\ref{pafinal}), (\ref{pafinal1}), (\ref{gafinal}), and (\ref{gafinal1}) once again, and Algorithm. \ref{Alg1} can be used as a solution to the optimization problem in (\ref{copt1}).

 In the next section, we present simulations to corroborate the results presented in this section, and in Theorem 1 and Theorem 2.
 \section{Results and Discussions}
In this section, we provide simulation results to highlight the validity of the analytical results presented in this work. {To this end, a room size of $5\text{m}\times5\text{m}\times3\text{m}$ was considered where the LED was located at a height of 2.25 metres. Four UEs considered at different random center-radii with respect to the LED transmitter are assumed to be mobile, with their corresponding channel gains given by (\ref{pdfeqn1})}. The modulation constellation is on-off keying (OOK), and the LED radiation pattern was assumed to be Lambertian. For power-allocation, generic gain-ratio power allocation \cite{marshoud2016non} was compared with the proposed power-allocation technique for performance evaluation.

In the first simulation setup, we validate the expressions for the p.d.f of $\phi$ presented in Theorem 1. The p.d.f of the additive distortion was estimated/simulated from the histogram of observations (the signal to noise ratio is greater than 10dB) in MATLAB using the ``histc" function. For illustration, the estimated p.d.f is shown in Fig. \ref{fig_0} for $\nu=10$. It can be observed that both approximations, namely: a) through Hermite polynomials (the approximation was constructed using only 10 Hermite polynomials), and b) from eq. (\ref{final1}) bear close overlap with the simulated histogram. In the considered setup $\beta$ is chosen as high as $\frac{\alpha}{3}$; yet it is found from the simulations that the derived analytical approximations presented in Theorem 1, well-model the simulated p.d.f.  

Next, in Fig. \ref{fig_2}, and Fig. \ref{fig_3}, the per-user capacity {for the static NOMA-VLC setup was simulated and compared with the derived formula for rate for LED half angles of $50^{\circ}$, and $60^{\circ}$ respectively}, for each user (termed as User-I, User-II, User-III and User-IV). It can be observed that the formula for capacity derived in (\ref{ag})  closely models the simulated per-user capacity. Furthermore, it is also  observed that the proposed power allocation achieves higher capacity for the ``premium" user as compared to classical gain ratio power allocation, whilst satisfying the rate-QoS constraint of each user for a  given power budget. Without loss of generality, the four users' rate-QoS is chosen as 0.2 bits per channel use (bpcu), 0.6 bpcu, 2 bpcu, and 5 bpcu ({please note that according to recent advances \cite{vaezi2018non}, simply allocating higher power to users with lower channel strengths does not guarantee the success of NOMA, rather it is the chosen rate-QoS points that make the case for NOMA in terms of sum-rate}).


Additionally, in Fig. \ref{fig_6}, Fig. \ref{fig_7}, Fig. \ref{fig_8} and Fig. \ref{fig_9} (which correspond to LED half angles of $50^{\circ}$ and $60^{\circ}$), we simulate the per-user capacity for the NOMA-VLC scenario impaired by user-mobility, wherein, each user was assumed to be mobile, with various values of $h_{\min}$ and $h_{\max}$ (which basically determine the extent of mobility/variability) \cite{yin2016performance}. In this scenario, the per-user capacity was simulated first, and validated with the corresponding analytical expression (\ref{srmobile}). It can be observed that the analytical expressions for per-user capacity bear a close match with simulations for all the considered users, which ratifies the presented analysis. 

Furthermore, the overall sum-rate for half angle of $50^{\circ}$ are plotted in Fig. \ref{fig_10} and Fig. \ref{fig_11} for both static and user-mobility scenarios respectively. It can be observed that there is a close agreement between analytical and simulated sum-rate curves dictated by (\ref{sr_mob}). Additionally, it can also be observed that the proposed power allocation minimizes the gap between the achieved sum-rate and the sum-rate for AWGN channel ({where there is no ambient light}), which makes the proposed power-allocation method viable. Lastly, a comparison is also provided by considering a hypothetical case when we optimize the sum of SH capacity formula for each user in (\ref{cop}) and (\ref{copt1}) instead of the capacity-formula derived in (\ref{rran}) ({in other words, interference due to ambient light exists, and instead of using the capacity formula derived in (\ref{rran}), the classical Shannon-Hartley capacity formula is deployed as an objective function in this comparison, and performance is benchmarked with the proposal}). It can be observed that there is a noticeable degradation in sum-rate performance if the effect of ambient light is neglected from the derivation of power-allocation coefficients, which makes the case for the power-allocation algorithm proposed in this work.

Lastly, it is worthwhile to mention that it is possible for GRPA to outperform the proposed power allocation in terms of capacity for some of the users. This can be attributed to the fact that our proposal ensures fulfilment of rate-QoS of each user with a given power-budget, and sum-rate objective. Hence, for a given user, GRPA can sometimes give an incremental benefit in terms of individual rate, as it is an unconstrained algorithm which is totally based on individual users' channel gains, and not on their individual QoS requirement. However, from the above figures, performance gains can be noted in terms of the achieved overall sum-rate, and in terms of the rate achieved by the ``premium" user, for ambient light impaired VLC channels. 
\section{Conclusion}
In this work, two major performance-inhibiting factors of VLC, namely, a) additive interference due to ambient light, and b) user mobility were addressed. The p.d.f of the additive interference was characterized, and the analytically derived p.d.f matches with the simulated p.d.f. Furthermore, a formula for channel-capacity is derived, and in cases with negligible ambient light, the derived formula for capacity is identical to the SH formula for AWGN channels. Further, an analytical expression for capacity was derived for user-mobility impaired NOMA-VLC channels, and verified by extensive simulations over realistic VLC scenarios. Lastly, an alternating projections based power allocation technique is derived for both static and user-mobility impaired NOMA-VLC scenarios with ambient lighting. From simulations, the proposed power allocation technique is found to achieve higher sum-rate as compared to classical gain-ratio power allocation, which makes it suitable for practical NOMA-VLC deployments.   

%

\ifCLASSOPTIONcaptionsoff
  \newpage
\fi

\bibliography{./paper}
\bibliographystyle{IEEEtran}

\begin{figure}[!htbp]
  \centering
  \includegraphics[width=\linewidth,height=13cm]{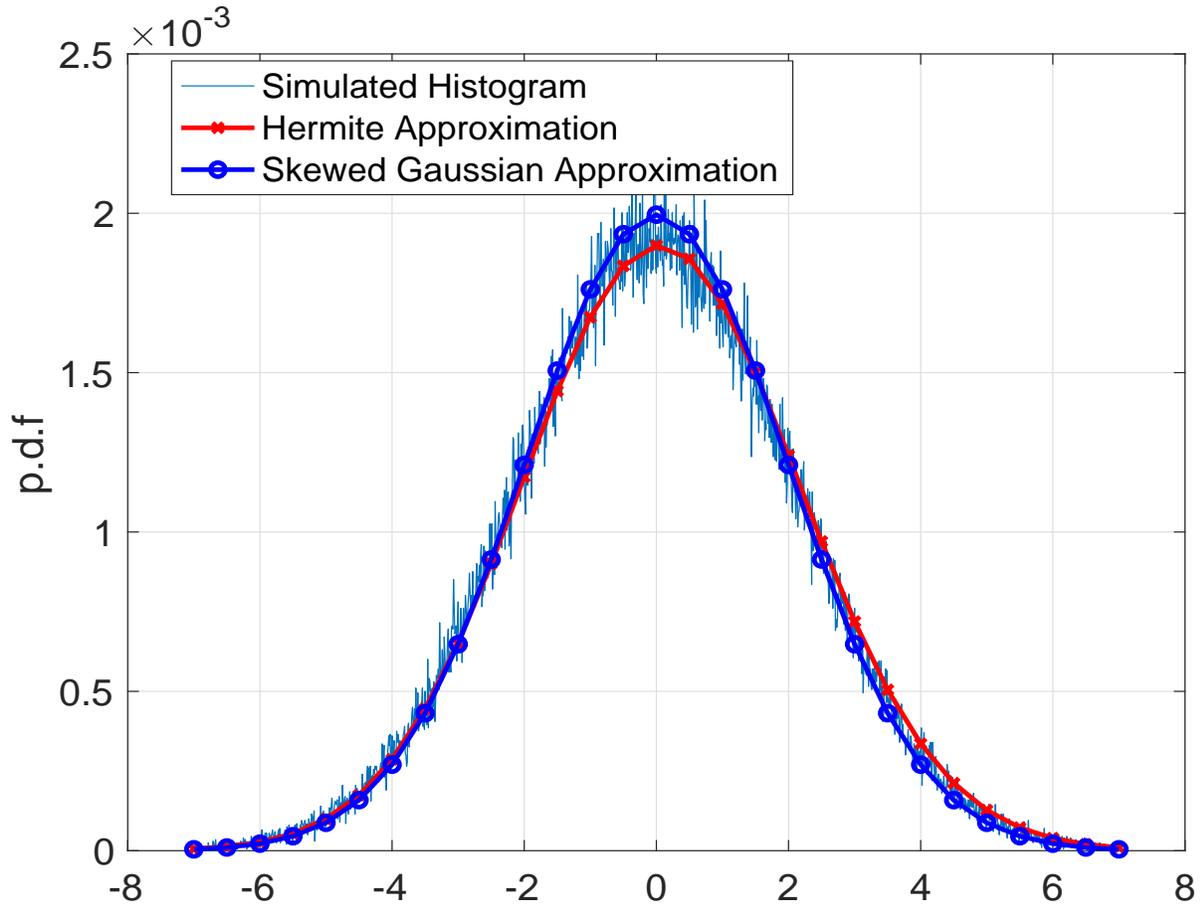}\\
  \caption{Validation of approximate expressions for p.d.f of the additive distortion consisting of AWGN and ambient light $\beta=\frac{\alpha}{3},\nu=10,\alpha=2$..}\label{fig_0}
\end{figure}
\begin{figure}
  \centering
  \includegraphics[width=\linewidth,height=13cm]{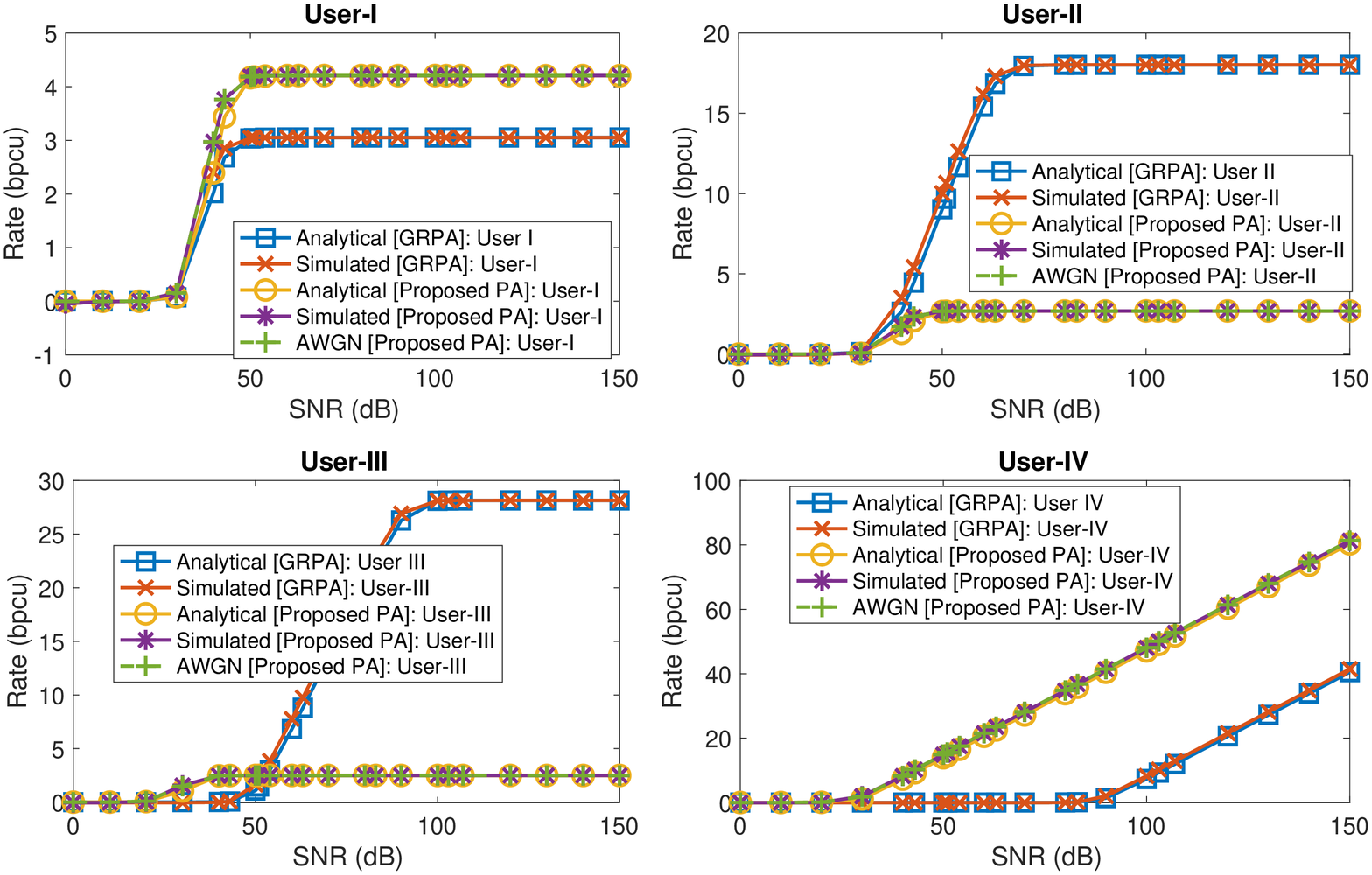}\\
  \caption{Validation of derived rate formula for $\beta=\frac{\alpha}{3},\nu=10,\alpha=2$ for static users for $\Theta_{\frac{1}{2}}=50^{\circ}$.}\label{fig_2}
\end{figure}
\begin{figure}
  \centering
  \includegraphics[width=\linewidth,height=13cm]{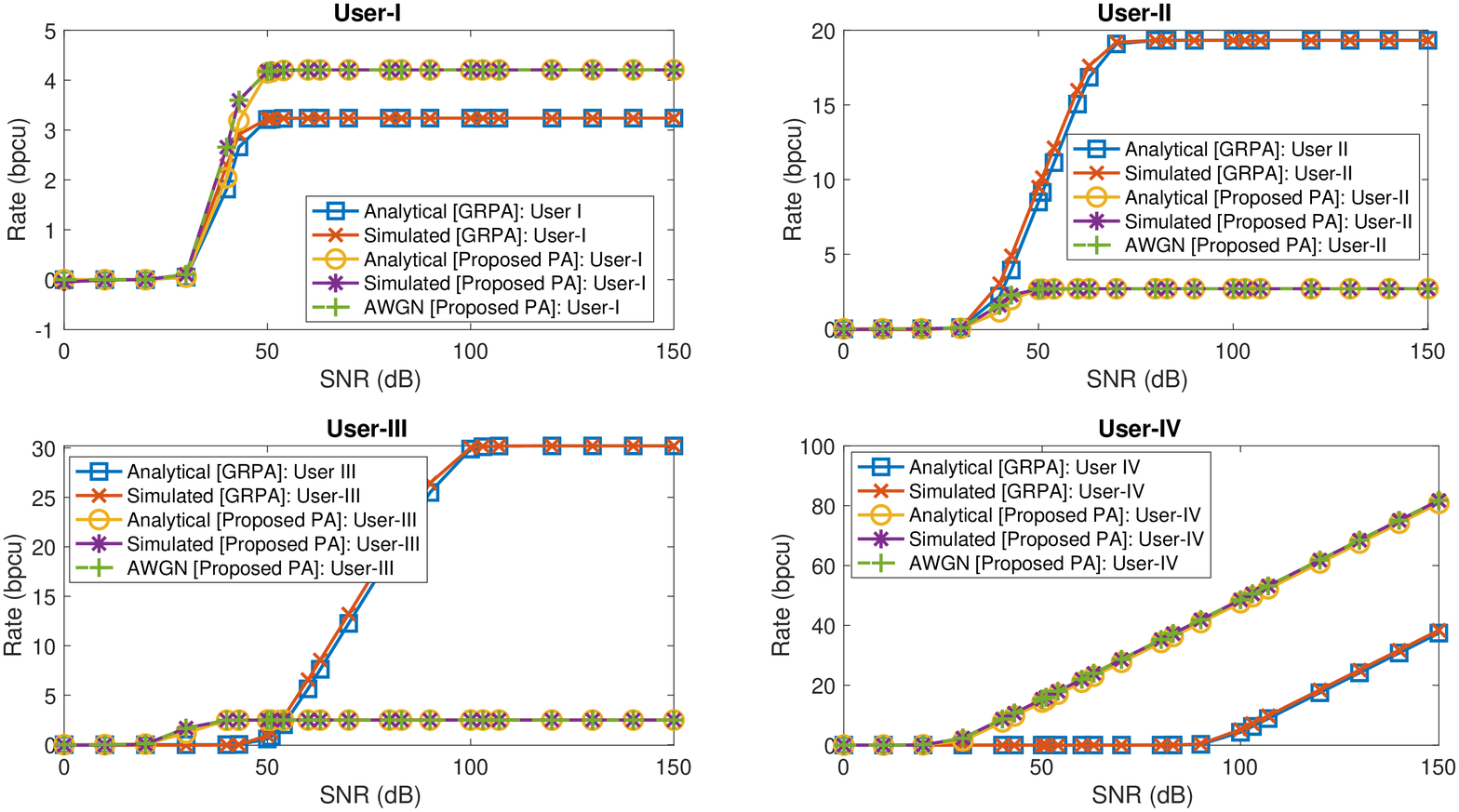}\\
  \caption{Validation of derived rate formula for $\beta=\frac{\alpha}{3},\nu=10,\alpha=2$ for static users for $\Theta_{\frac{1}{2}}=60^{\circ}$.}\label{fig_3}
\end{figure}
\begin{figure}
  \centering
  \includegraphics[width=\linewidth,height=16cm]{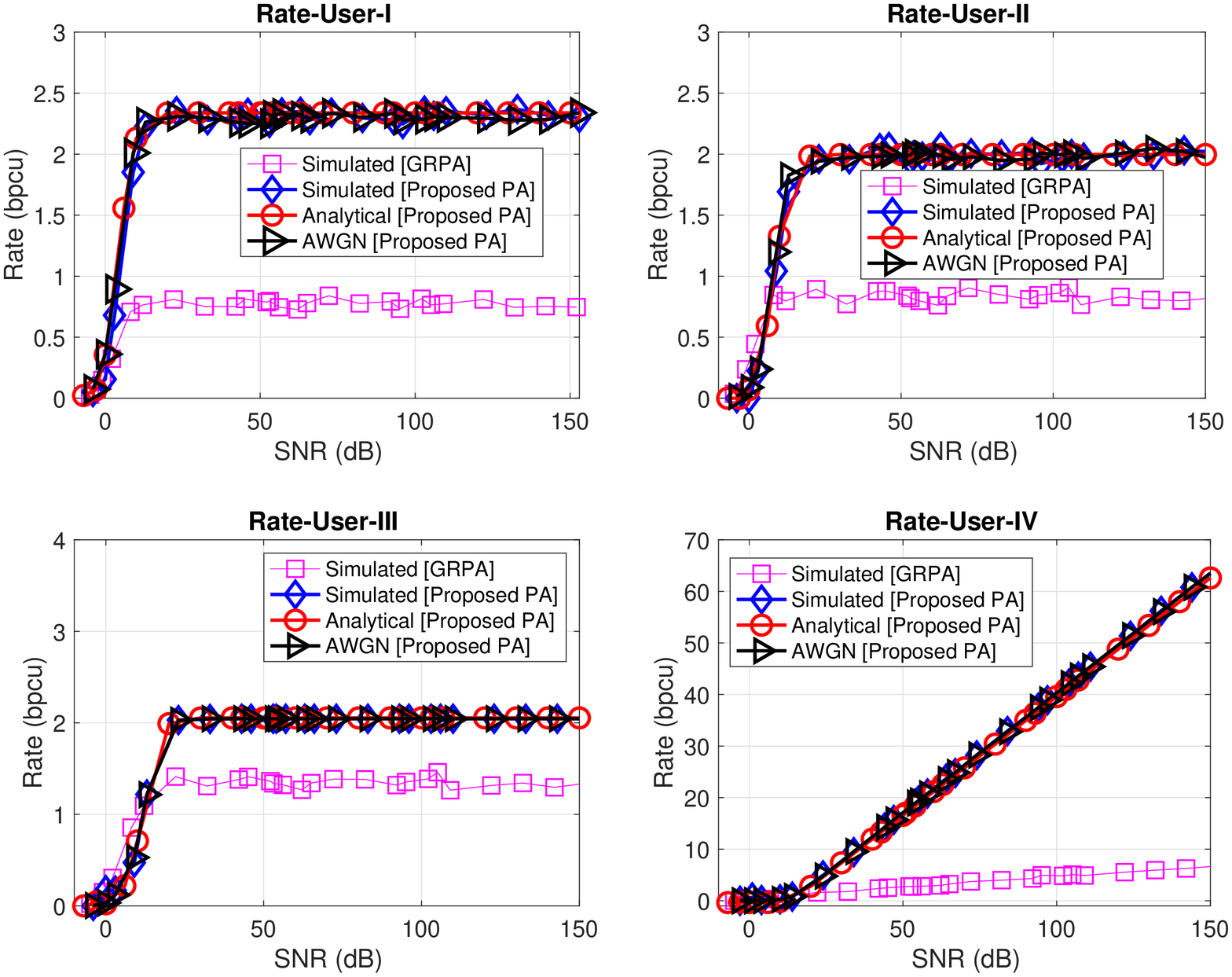}\\
  \caption{Validation of derived rate formula in presence of mobility for $\beta=\frac{\alpha}{3},\nu=10,\alpha=2$, $h_{\min}=1$, $h_{\max}=3$, $\Theta_{\frac{1}{2}}=50^{\circ}$.}\label{fig_6}
\end{figure}
\begin{figure}
  \centering
  \includegraphics[width=\linewidth,height=16cm]{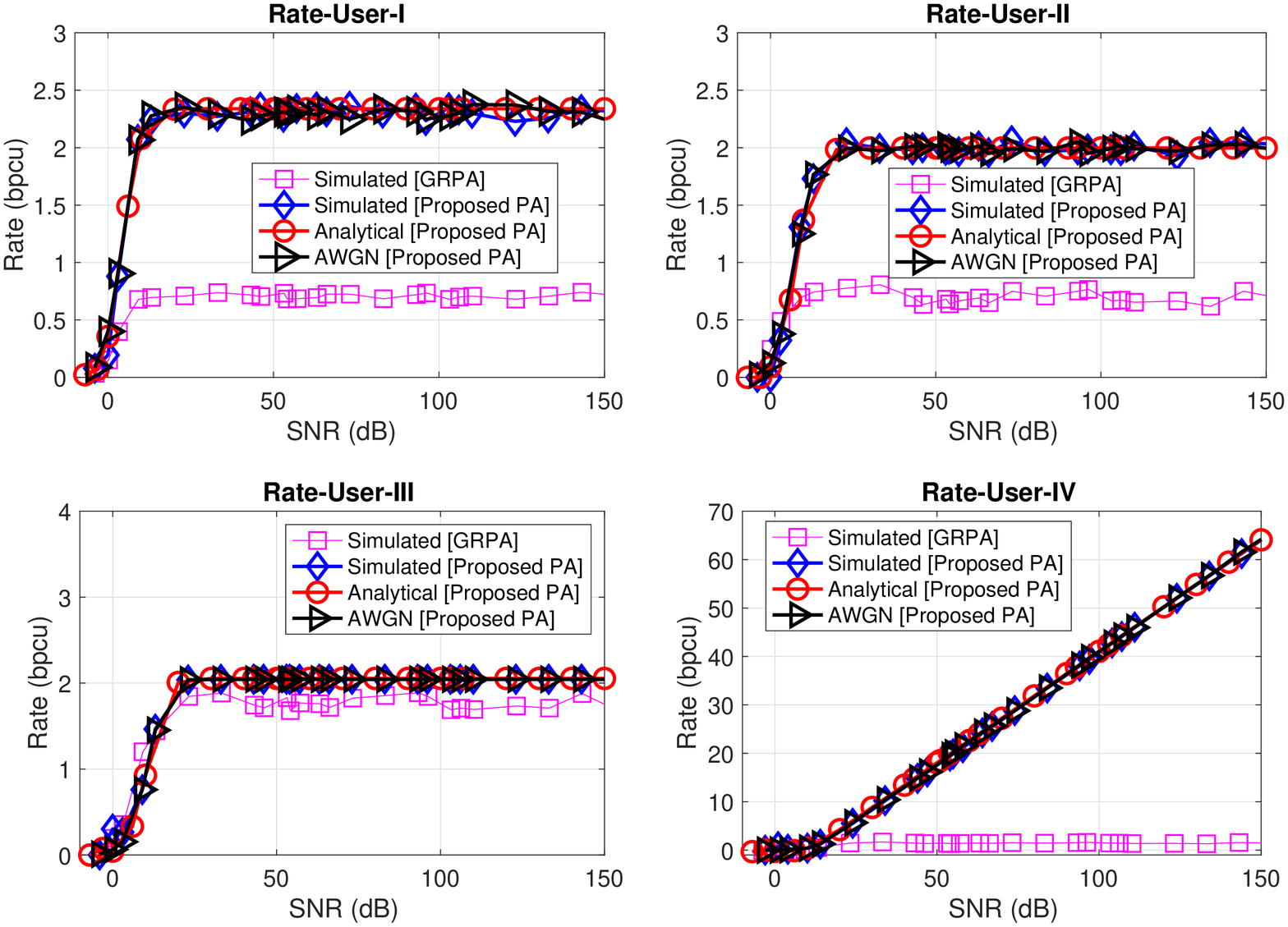}\\
  \caption{Validation of derived rate formula in presence of mobility for $\beta=\frac{\alpha}{3},\nu=10,\alpha=2$, $h_{\min}=1$, $h_{\max}=5$, $\Theta_{\frac{1}{2}}=50^{\circ}$.}\label{fig_7}
\end{figure}
\begin{figure}
  \centering
  \includegraphics[width=\linewidth,height=15cm]{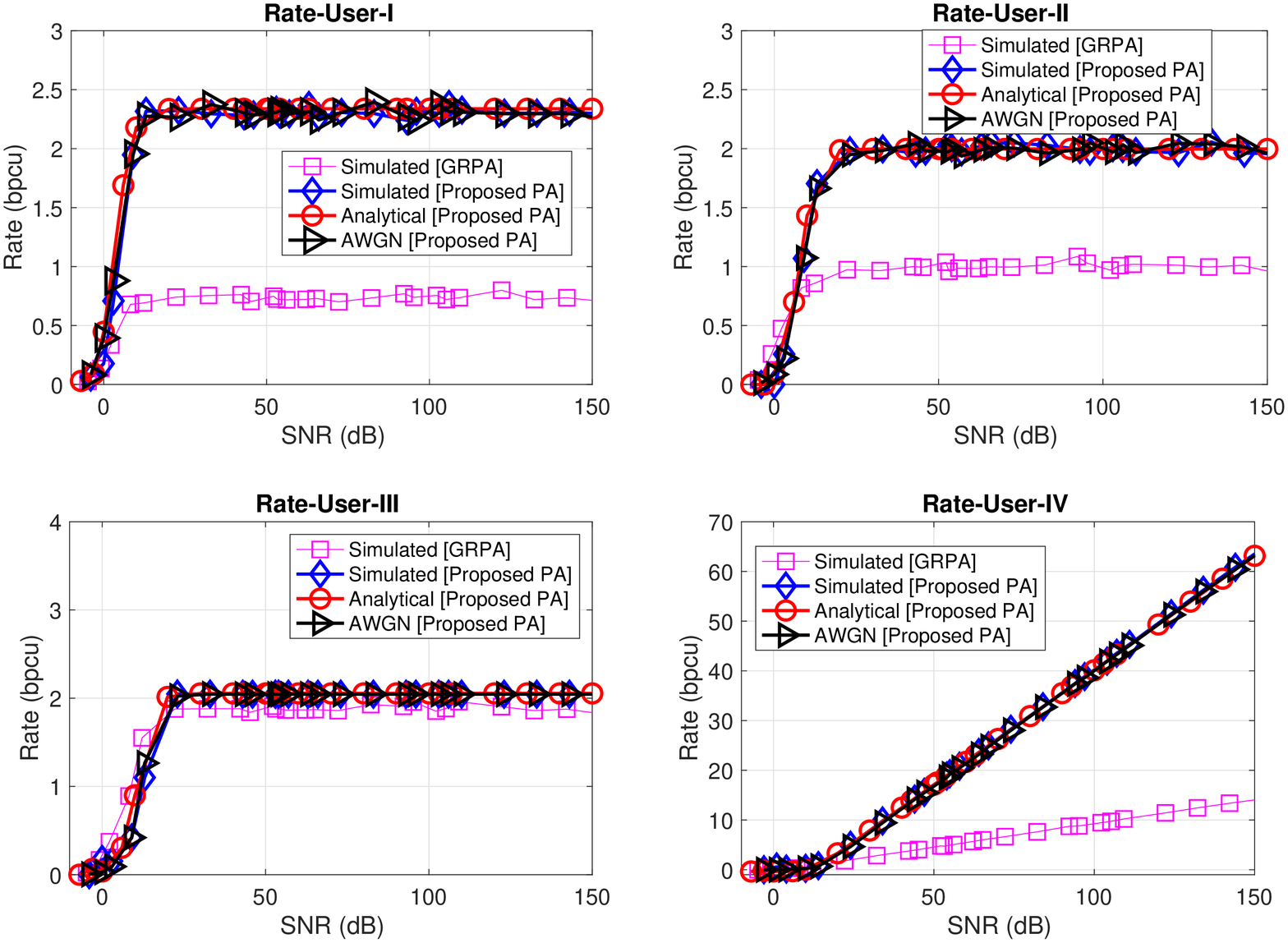}\\
  \caption{Validation of derived rate formula in presence of mobility for $\beta=\frac{\alpha}{3},\nu=10,\alpha=2$, $h_{\min}=1$, $h_{\max}=3$, $\Theta_{\frac{1}{2}}=60^{\circ}$.}\label{fig_8}
\end{figure}
\begin{figure}
  \centering
  \includegraphics[width=\linewidth,height=15cm]{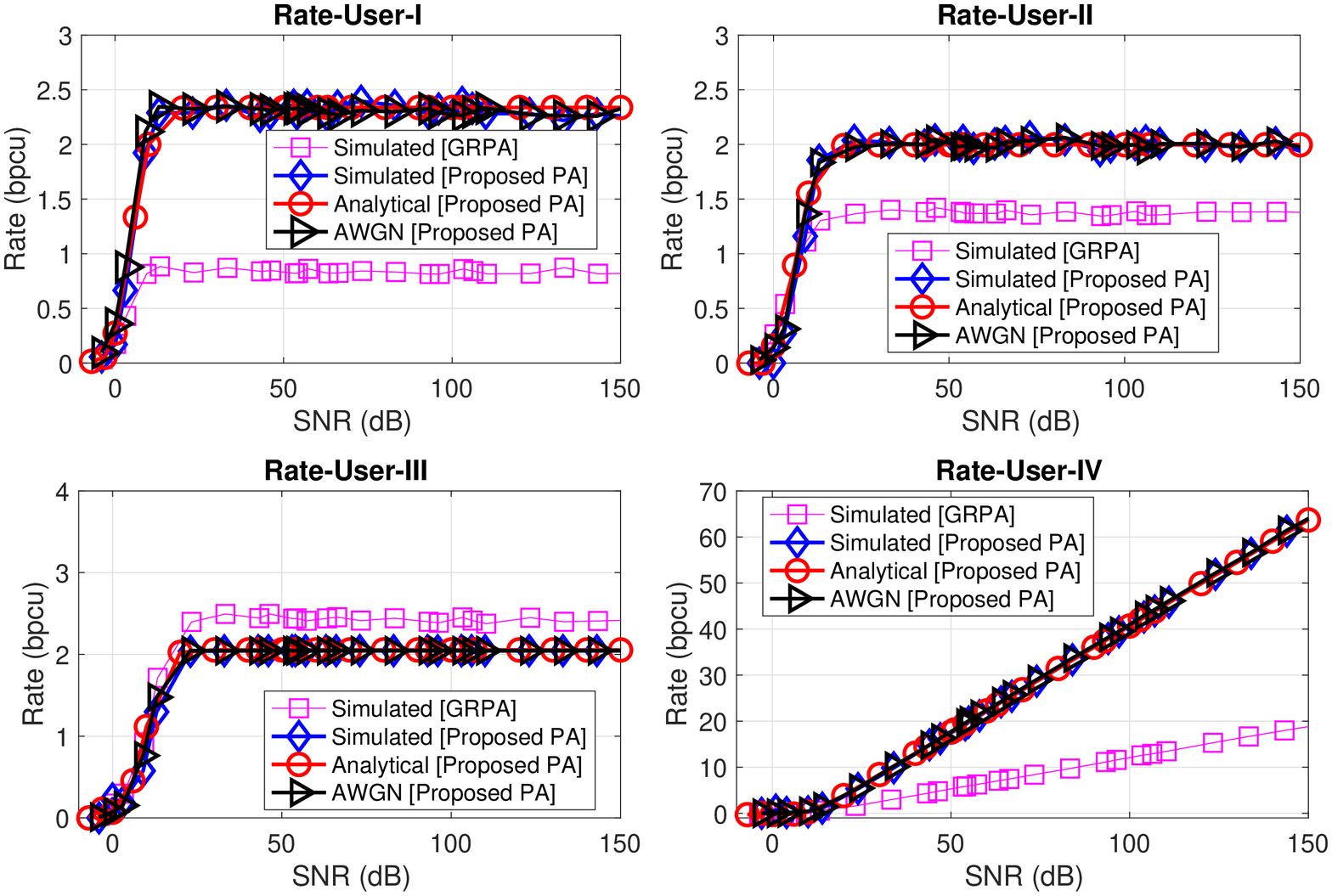}\\
  \caption{Validation of derived rate formula in presence of mobility for $\beta=\frac{\alpha}{3},\nu=10,\alpha=2$, $h_{\min}=1$, $h_{\max}=5$, $\Theta_{\frac{1}{2}}=60^{\circ}$.}\label{fig_9}
\end{figure}
\begin{figure}
  \centering
  \includegraphics[width=15cm,height=10cm]{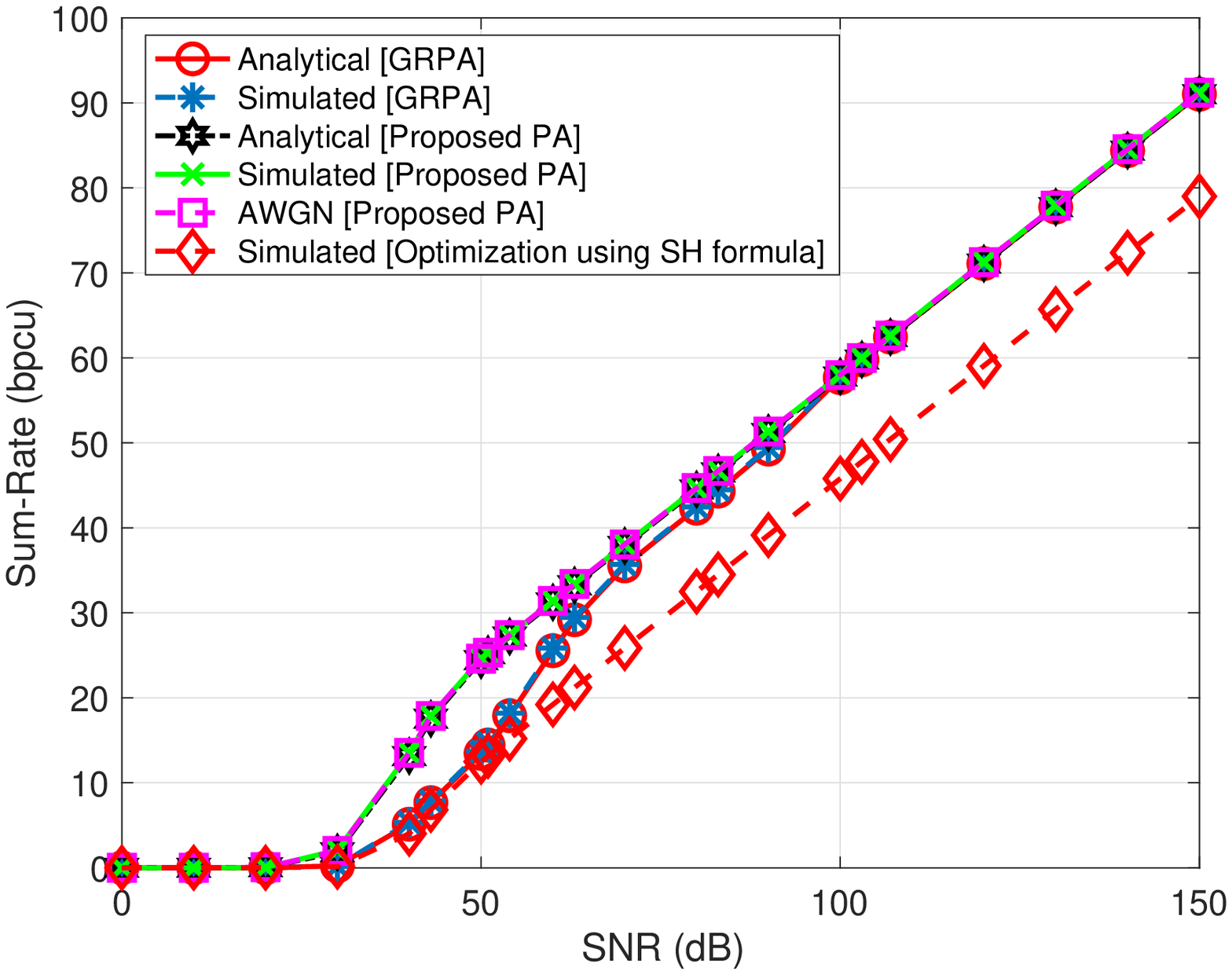}\\
  \caption{Validation of sum-rate for $\beta=\frac{\alpha}{3},\nu=10,\alpha=2$ for static users for $\Theta_{\frac{1}{2}}=50^{\circ}$.}\label{fig_10}
\end{figure}
\begin{figure}
  \centering
  \includegraphics[width=15cm,height=10cm]{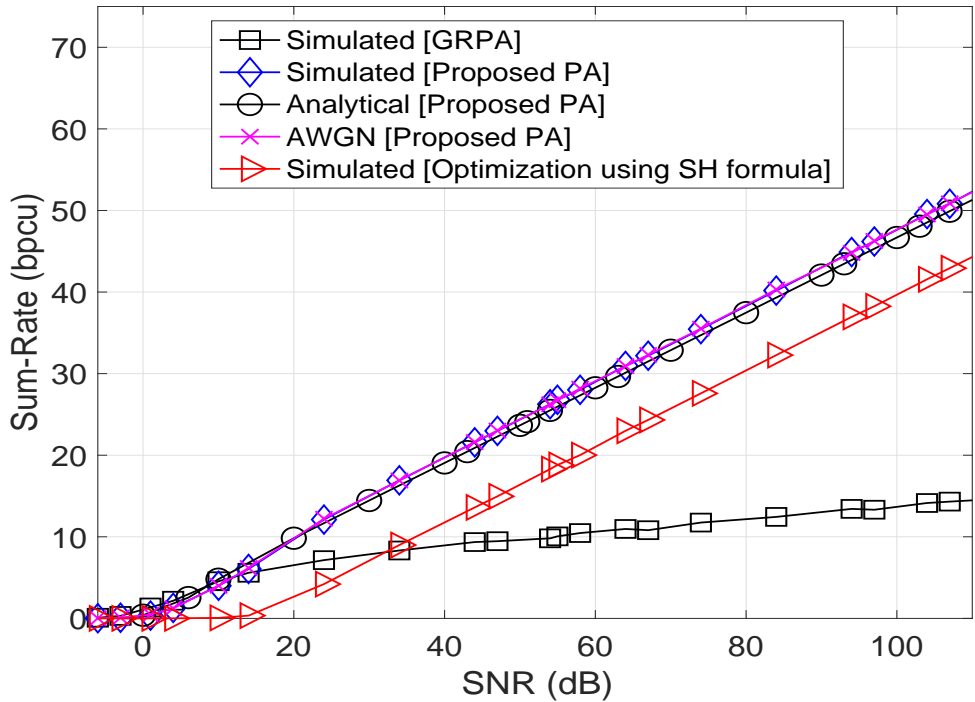}\\
  \caption{Validation of derived rate formula for $\beta=\frac{\alpha}{3},\nu=10,\alpha=2$ for scenario with user-mobility for $\Theta_{\frac{1}{2}}=50^{\circ}$.}\label{fig_11}
\end{figure}

\end{document}